\documentclass[a4paper,twoside]{article}

\usepackage[dvips]{graphicx}
\usepackage{color}
\usepackage[utf8]{inputenc}
\usepackage[T2A,TS1,T1]{fontenc}
\usepackage[english]{babel}
\usepackage{amsmath,amsthm,amssymb}
\usepackage{euscript}
\usepackage{calc}

\PassOptionsToPackage{lowtilde}{url}
\RequirePackage[%
    colorlinks=true,
    pdfstartview=FitH,
    linkcolor=blue,
    citecolor=blue,
    filecolor=blue,
    urlcolor=blue,
    linktoc=page,
    pdftitle={On Asymptotic Gate Complexity and Depth of Reversible Circuits Without Additional Memory},
    pdfauthor={Dmitry V. Zakablukov},
    pdfkeywords={reversible logic, gate complexity, circuit depth, asymptotic bounds},
    pdfsubject={The paper discusses the asymptotic gate complexity and depth of a reversible circuit consisting of NOT, CNOT and 2-CNOT gates}]{hyperref}

\newcommand{\NN}{\mathbb N}
\newcommand{\ZZ}{\mathbb Z}
\newcommand{\frS}{\mathfrak S}
\newcommand{\vv}{\mathbf}
\newcommand{\LC}{L^\text{(C)}}
\newcommand{\LT}{L^\text{(T)}}
\newcommand{\WC}{W^\text{(C)}}
\newcommand{\WT}{W^\text{(T)}}

\title{On Asymptotic Gate Complexity and Depth of Reversible Circuits Without Additional Memory}
\author{{
    {\small
        \begin{tabular}[t]{c}
        \normalsize Dmitry V. Zakablukov\medskip\\
        Dep. of Information Security, Bauman Moscow State Technical University,\\
        Moscow, Russian Federation\smallskip\\
        E-mail: \texttt{\href{mailto:dmitriy.zakablukov@gmail.com}{dmitriy.zakablukov@gmail.com}}
        \end{tabular}
    }
    }
}

\theoremstyle{plain}
\newtheorem{theorem}{Theorem}

\theoremstyle{definition}
\newtheorem{definition}{Definition}

\begin{document}

\maketitle

\begin{abstract}
Reversible computation is one of the most promising emerging technologies of the future.
The usage of reversible circuits in computing devices can lead to a significantly lower power consumption.
In this paper we study reversible logic circuits consisting of NOT, CNOT and 2-CNOT gates.
We introduce a set $F(n,q)$ of all transformations $\ZZ_2^n \to \ZZ_2^n$ that can be implemented by reversible circuits
with $(n+q)$ inputs.
We define the Shannon gate complexity function $L(n,q)$ and the depth function $D(n,q)$ as functions of $n$ and the number
of additional inputs $q$.
First, we prove general lower bounds for functions $L(n,q)$ and $D(n,q)$.
Second, we introduce a new group theory based synthesis algorithm, which can produce a circuit $\frS$ without additional inputs
and with the gate complexity $L(\frS) \leq 3n 2^{n+4}(1+o(1)) \mathop / \log_2 n$.
Using these bounds, we state that almost every reversible circuit with no additional inputs,
consisting of NOT, CNOT and 2-CNOT gates, implements a transformation from $F(n,0)$ with the gate complexity
$L(n,0) \asymp n 2^n \mathop / \log_2 n$ and with the depth $D(n,0) \geq 2^n(1-o(1)) \mathop / (3\log_2 n)$.
\end{abstract}


\section{Introduction}

Reversible logic is essential in quantum computing, but it also has a great potential in designing
various computing devices with low power consumption. Landauer proved~\cite{landauer} that irreversible
computations lead to energy dissipation regardless of the underlying technology.
Moreover, Bennett showed~\cite{bennett} that zero-level of energy loss can be achieved only when a circuit
is completely built from reversible gates.
The main problem is that reversible circuits with fewer number of gates (\textit{gate complexity})
and input count are more practical to use. Unfortunately, strict asymptotic bounds
for the gate complexity of reversible circuits haven't been found so far.

Circuit complexity theory goes back to the work of Shannon~\cite{shannon}. He suggested considering
a complexity of the minimal switching circuit implementing some Boolean function as a measure of complexity of this function.
For today, the asymptotic gate complexity $L(n) \sim 2^n \mathop / n$ of a Boolean function of $n$ variables in a basis
of classical gates ``NOT, OR, AND'' is well-known.

Reversible computations were discussed by Toffoli in 1980~\cite{toffoli}. He described the first reversible gate,
2-CNOT (controlled controlled NOT). After that various reversible gates (CNOT~\cite{feynman}, Fredkin, etc.)
were introduced. The subject of this paper is reversible logic circuits consisting of NOT, CNOT and 2-CNOT gates.
A formal definition of these gates from~\cite{my_fast_synthesis_algorithm} will be used.
It is well known that any even permutation $h \in A(\ZZ_2^n)$ can be implemented in a circuit with $n$ inputs,
consisting of NOT, CNOT and 2-CNOT gates~\cite{shende}. Hence, the gate complexity or the depth of this circuit
can be considered as a measure of the permutation $h$ complexity.

In this paper we describe a set $F(n,q)$ of all transformations $\ZZ_2^n \to \ZZ_2^n$ that can be implemented
by reversible circuits with $(n+q)$ inputs.
We estimate the gate complexity and the depth of reversible circuit, implementing some transformation $f \in F(n,q)$ with $q$
additional inputs (also referred to as an additional memory).
For this purpose we define the Shannon gate complexity function $L(n,q)$ and the depth function $D(n,q)$
as functions of $n$ and the number of additional inputs $q$.

Using the counting argument, we prove general lower bounds for the functions $L(n,q)$ and $D(n,q)$:
\begin{gather*}
    L(n,q) \geq \frac{n2^n}{3\log_2(n+q)}(1 - o(1)) \;  , \\
    D(n,q) \geq \frac{n2^n}{3(n+q)\log_2(n+q)}(1 - o(1)) \;  .
\end{gather*}

After that we introduce a new group theory based synthesis algorithm,
which can produce a circuit $\frS$ without additional inputs and with the gate complexity
$L(\frS) \leq 3n 2^{n+4}(1+o(1)) \mathop / \log_2 n$ and the depth $D(\frS) \leq n2^{n+5}(1+o(1)) \mathop / \log_2 n$.
Finally, using these lower and upper bounds, we formulate the main statement of this paper:
almost every reversible circuit with no additional inputs, consisting of NOT, CNOT and 2-CNOT gates,
implements a transformation from $F(n,0)$ with the gate complexity
$L(n,0) \asymp n 2^n \mathop / \log_2 n$ and with the depth $D(n,0) \geq 2^n(1-o(1)) \mathop / (3\log_2 n)$.

\section{Background}

The concept of reversible gates was discussed by Toffoli in 1980~\cite{toffoli}.
Gates NOT and $k$-CNOT and the synthesis of circuits consisting of these gates
were discussed, for example, in~\cite{my_fast_synthesis_algorithm}.
We will use the following formal definitions of NOT and $k$-CNOT gates.
\begin{definition}\label{formula_not_definition}
    Gate $N_j^n$ is a NOT gate with $n$ inputs, which defines the transformation $f_j\colon \ZZ_2^n \to \ZZ_2^n$ as follows:
    $$
        f_j(\langle x_1, \cdots, x_j, \cdots, x_n \rangle) =
            \langle x_1, \cdots, x_j \oplus 1, \cdots, x_n \rangle \;  .
    $$
\end{definition}
\begin{definition}\label{formula_k_cnot_definition}
    Gate $C_{i_1,\cdots,i_k;j}^n = C_{I;j}^n$, $j \notin I$, is a generalized Toffoli gate ($k$-CNOT) with $n$ inputs,
    $k$ control inputs, which defines the transformation $f_{I;j}\colon \ZZ_2^n \to \ZZ_2^n$ as follows:
    $$
        f_{I;j}(\langle x_1, \cdots, x_j, \cdots, x_n \rangle) =
            \langle x_1, \cdots, x_j \oplus x_{i_1} \wedge \cdots \wedge x_{i_k}, \cdots, x_n \rangle \;  .
    $$
\end{definition}
We will omit an upper index in $N_j^n$ and $C_{i_1,\cdots,i_k;j}^n$, if the value of $n$ is clear from the context.
Also we will refer to $N_j$ and $C_{i_1,\cdots,i_k;j}$ as $TOF(j)$ and $TOF(i_1,\cdots,i_k;j)=TOF(I;j)$, respectively.
It is obvious that in this case the equality $TOF(j) = TOF(\emptyset; j)$ holds.

Let's denote a set of all NOT, CNOT (Feynman) and 2-CNOT (Toffoli) gates with $n$ inputs as $\Omega_n^2$.

A circuit of gates is usually defined as an acyclic oriented graph with marked edges and vertices.
In case of reversible circuits of gates from $\Omega_n^2$, fan-in, fan-out and random connection of inputs and outputs
of gates are forbidden. In an oriented graph describing a reversible circuit $\frS$, all the vertices corresponding to gates
have exactly $n$ numbered inputs and outputs. These vertices are numbered from $1$ to $l$ and $i$-th output of $m$-th vertex,
$m < l$, is connected only to an $i$-th input of $(m+1)$-th vertex.
The circuit inputs are the inputs of the first vertex and the circuit outputs are the outputs of the $l$-th vertex.
We will also call such a connection of gates as \textit{composition}.

For every vertex in the graph, $i$-th input and output are assigned to a symbol $r_i$ from some set
$R = \{\,r_1, \cdots, r_n\,\}$.
All symbols $r_i$ can be treated as memory registers names (memory cells indices), storing the current computation result
of the circuit. From definitions~\eqref{formula_not_definition} and~\eqref{formula_k_cnot_definition} it follows that
the value of only one memory register can be inverted at a time. This makes an essential difference
between reversible circuits and irreversible ones.

Among all the properties of a reversible circuit the most important ones for us are the gate complexity and the depth.
Let a reversible circuit $\frS$ with $n$ inputs be a composition of $l$ gates from $\Omega_n^2$:
$\frS = \mathop{*}_{j=1}^l {TOF(I_j; t_j)}$, where $t_j$ and $I_j$ are the controlled output and the set of control inputs
of $j$-th gate respectively.
\begin{definition}
    The gate complexity $L(\frS)$ of the reversible circuit $\frS = \mathop{*}_{j=1}^l {TOF(I_j; t_j)}$ is the number of 
    gates $l$.
\end{definition}

Classically a circuit's depth is defined as the length of the longest path from an input to an output vertex of the graph,
associated with this circuit. In our model of a reversible circuit, the associated graph presents itself a single chain,
so if we use a classical definition of a circuit's depth, we will get it equal to the circuit's gate complexity.
But it is clear that in reality it is not the case. To keep our reversible circuit's model, we introduce an alternative,
but equivalent definition of a reversible circuit's depth.
\begin{definition}
    Reversible circuit $\frS = \mathop{*}_{j=1}^l {TOF(I_j; t_j)}$ has depth $D(\frS) = 1$, if for every two of its gates
    $TOF(I_1; j_1)$ and $TOF(I_2; j_2)$ the following equation holds:
    $$
        \left( \{\,t_1\,\} \cup I_1 \right) \cap \left( \{\,t_2\,\} \cup I_2 \right) = \emptyset \;  .
    $$
\end{definition}
\begin{definition}
    Reversible circuit $\frS$ has depth $D(\frS) \leq d$,
    if it can be divided into $d$ disjoint sub-circuits with the depth of each equal to 1:
    \begin{equation}
        \frS = \bigsqcup_{i=1}^d{{\frS'}_i}, \text{ } {\frS'}_i \subseteq \frS, \text{ }D({\frS'}_i) = 1 \;  .
        \label{formula_decomposition_of_scheme_for_depth}
    \end{equation}
\end{definition}

Now we can rigorously define a reversible circuit's depth.
\begin{definition}\label{define_circuit_depth}
    The depth $D(\frS)$ of a reversible circuit $\frS$ is the minimal number of disjoint sub-circuits with the depth of each
    equal to 1 from the equation~\eqref{formula_decomposition_of_scheme_for_depth}.
\end{definition}

From the Definition~\ref{define_circuit_depth} we can derive a simple equation for the depth function in case
of a reversible circuit $\frS$ with $n$ inputs:
\begin{equation}\label{formula_depth_and_gate_complexity_dependency}
    L(\frS) \mathop / n \leq D(\frS) \leq L(\frS) \;  .    
\end{equation}

For example, let's consider a reversible circuit $\frS = C_{1;2} * C_{3;1} * N_2 * N_4 * C_{1,4;2} * N_3$
(see Fig.~\ref{pic_scheme_example}).
The circuit has six gates, so its gate complexity is $L(\frS) = 6$. Also, we can divide the circuit into 3 disjoint
sub-circuits with the depth of each equal to 1:
$\frS = (C_{1;2}) * (C_{3;1} * N_2 * N_4) * (C_{1,4;2} * N_3)$. So the circuit's depth is $D(\frS) = 3$.

\begin{figure}[ht]
    \centering
    \includegraphics{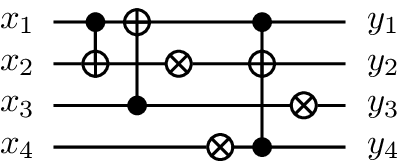}
    \caption
    {
        A reversible circuit $\frS = C_{1;2} * C_{3;1} * N_2 * N_4 * C_{1,4;2} * N_3$
        with the gate complexity $L(\frS) = 6$ and the depth $D(\frS) = 3$.
    }\label{pic_scheme_example}
\end{figure}

From Fig.~\ref{pic_scheme_example} one can note that our reversible circuit is equivalent to another one
with the depth equal to 3:
$\frS_1 = (C_{1;2} * N_4) * (C_{3;1} * N_2) * (C_{1,4;2} * N_3)$. Therefore from here on we will consider,
that such circuits $\frS$ and $\frS_1$ are different in terms of our reversible circuit's model,
but equivalent in terms of the equality of Boolean transformations, defined by them.

\section{Shannon gate complexity, depth and quantum weight functions}

It was proved that a reversible circuit with $n \geq 4$ inputs defines an even permutation on the set $\ZZ_2^n$ \cite{shende}.
In the same time, it can implement a transformation $\ZZ_2^m \to \ZZ_2^k$, where $m, k \leq n$,
with or without additional inputs. We need the following functions to explain this:
\begin{itemize}
    \item
        \textit{expanding} function $\phi_{n,n+k}\colon \ZZ_2^n \to \ZZ_2^{n+k}$ defined as
        $$
            \phi_{n,n+k}( \langle x_1, \cdots, x_n \rangle ) = \langle x_1, \cdots, x_n, 0, \cdots, 0 \rangle \;  .
        $$

    \item
        \textit{reducing} function $\psi_{n+k,n}^\pi\colon \ZZ_2^{n+k} \to \ZZ_2^n$ defined as
        $$
            \psi_{n+k,n}^\pi( \langle x_1, \cdots, x_{n+k} \rangle ) = \langle x_{\pi(1)}, \cdots, x_{\pi(n)} \rangle \;  ,
        $$
        where $\pi$ is a permutation on the set $\ZZ_{n+k}$.
\end{itemize}

Let us now define a reversible circuit implementing a transformation (see Fig.~\ref{pic_function_realization_with_memory}).
\begin{figure}[ht]
    \centering
    \includegraphics{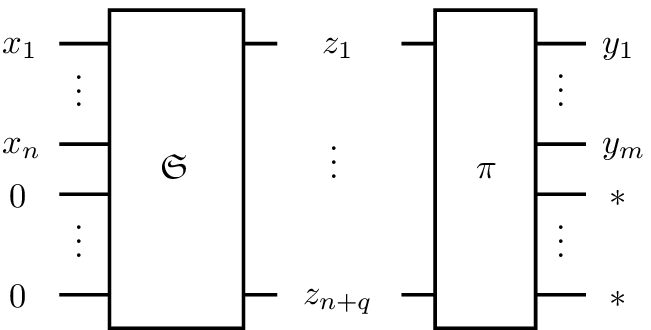}
    \caption
    {
        A reversible circuit $\frS$ implementing a transformation $f\colon \ZZ_2^n \to \ZZ_2^m$
        with $q$ additional inputs. For every $\vv x \in \ZZ_2^n$ the equation
        $f(\langle x_1, \cdots, x_n\rangle) = \langle y_1, \cdots, y_m\rangle$ holds.
    }\label{pic_function_realization_with_memory}
\end{figure}

\begin{definition}
    A reversible circuit $\frS_g$ with $(n+q)$ inputs, defining a transformation
    $g\colon \ZZ_2^{n+q} \to \ZZ_2^{n+q}$,
    implements a transformation $f\colon \ZZ_2^n \to \ZZ_2^n$ using $q \geq 0$ additional inputs
    (additional memory), if there is such a permutation $\pi \in S(\ZZ_{n+q})$ that for every $\vv x \in \ZZ_2^n$
    the following equation holds:
    $$
        \psi_{n+q,n}^\pi(g( \phi_{n,n+q}(\vv x))) = f(\vv x) \;  .
    $$
\end{definition}

Note that in this terminology expressions ``implements a transformation'' and ``defines a transformation'' have
different meanings: if a reversible circuit $\frS_g$ defines a transformation $f$, then $g(\vv x) = f(\vv x)$ for all $\vv x$.
If a circuit $\frS_g$ implements a transformation $f\colon \ZZ_2^n \to \ZZ_2^n$ and has exactly $n$ inputs,
we will say that this circuit implements $f$ \textit{without additional inputs}.

Let $P_2(n,n)$ be the set of all transformations $\ZZ_2^n \to \ZZ_2^n$.
Let $F(n,q) \subseteq P_2(n,n)$ be the set of all transformations that can be implemented by reversible circuits with
$(n+q)$ inputs.
The set of permutations, corresponding to all the gates from $\Omega_n^2$, generates the alternating group $A(\ZZ_2^n)$
and the symmetric group $S(\ZZ_2^n)$ for $n > 3$ and $n \leq 3$ respectively~\cite{shende}.
This implies that $F(n,0)$ is equal to the set of transformations that are defined by all the permutations from $A(\ZZ_2^n)$
and from $S(\ZZ_2^n)$ for $n > 3$ and $n \leq 3$ respectively.
On the other hand, it is not difficult to show that for $q \geq n$ the equality $F(n,q) = P_2(n,n)$ holds.

Let's consider a transformation $f \in F(n,q)$. Among all reversible circuits, consisting of gates from $\Omega_{n+q}^2$
and implementing the transformation $f$ with $q$ additional inputs, we can find a circuit $\frS_l$
with the minimum gate complexity and a circuit $\frS_d$ with the minimum depth.
Let $L(f,q) = L(\frS_l)$ and $D(f,q) = D(\frS_d)$.
Now we can define the Shannon gate complexity function $L(n,q)$ and the depth function $D(n,q)$ as follows:
\begin{gather}
    L(n,q) = \max_{f \in F(n,q)} {L(f,q)} \;  , \\
    D(n,q) = \max_{f \in F(n,q)} {D(f,q)} \;  .
\end{gather}

If we consider all the gates from $\Omega_n^2$ regardless of an underlying technology, we can assume that they all have the same
technological cost. However, in a quantum technology, for example, a technological cost of NOT and CNOT gates is much less than
a technological cost of a Toffoli gate~\cite{barenco}.
Hence, we will assume that a gate $e$ from $\Omega_n^2$ has the weight $W(e)$ depending on the underlying technology.

We define a \textit{quantum weight} function $W(\frS)$ for a reversible circuit $\frS$
as a sum of weights of all its gates.
Note that the value of $W(\frS)$ is not equal to the technological cost of a reversible circuit $\frS$,
because they may significantly differ.
But we can state that in most cases a greater value of the function $W(\frS)$ means a greater
technological cost of a reversible circuit $\frS$.

Let's define the function $W(f,q)$ in a similar way as the functions $L(f,q)$ and $D(f,q)$.
Then we can define the Shannon quantum weight function $W(n,q)$ as follows:
\begin{equation}
    W(n,q) = \max_{f \in F(n,q)} {W(f,q)} \;  .
\end{equation}

Let's also assume that all NOT and CNOT gates from $\Omega_n^2$ have the same weight $\WC$ and
all 2-CNOT gates from $\Omega_n^2$ have weight $\WT$.
If we denote the number of NOT and CNOT gates in a reversible circuit $\frS$ as $\LC(\frS)$
and the number of 2-CNOT gates as $\LT(\frS)$, then we can derive a simple equality for the quantum weight function
\begin{equation}
    W(\frS) = \WC \cdot \LC(\frS) + \WT \cdot \LT(\frS) \;  .
    \label{formula_quantum_weight_and_gate_complexity_dependency}
\end{equation}
Equation~\eqref{formula_quantum_weight_and_gate_complexity_dependency} means that we should count the number of 2-CNOT gates
in a reversible circuit separately from the other ones.

Many reversible logic synthesis algorithms were proposed recently~\cite{my_fast_synthesis_algorithm, iterative_compositions,
miller_spectral, miller_transform_based, saeedi_novel, maslov_rm_synthesis, saeedi_cycle_based}.
For almost every one of them an upper bound for the gate complexity of a synthesized circuit is proved.
The best known is the upper bound $L(\frS) \lesssim 5n 2^n$ for a reversible circuit $\frS$ without additional inputs,
consisting of gates from $\Omega_n^2$ \cite{maslov_rm_synthesis}.
We can consider this bound as the best upper bound for the function $L(n,0)$:
$$
    L(n,0) \lesssim 5n2^n \;  .
$$

Unfortunately, there are no known general lower bounds for the functions $L(n,q)$ and $D(n,q)$ for today.
In~\cite{shende} a lower bound $\Omega(n2^n \mathop / \log n)$ for the function $L(n,0)$ was proved.
In~\cite{maslov_thesis} a lower bound for the gate complexity of a reversible circuit without additional inputs,
consisting of gates mEXOR, was proved. However, the gate complexity of reversible circuits with additional inputs
was out of the scope.

The main result of this paper is the following theorems.
\begin{theorem}\label{theorem_complexity_lower}
    There is such $n_0 \in \NN$ that for $n > n_0$ the following equation holds:
    $$
        L(n,q) \geq \frac{2^n(n-2)}{3\log_2 (n+q)} - \frac{n}{3} \;  .
    $$
\end{theorem}
The proof of the Theorem~\ref{theorem_complexity_lower} will be given in Section~\ref{section_lower_bounds}.

\begin{theorem}\label{theorem_depth_lower}
    There is such $n_0 \in \NN$ that for $n > n_0$ the following equation holds:
    $$
        D(n,q) \geq \frac{2^n(n-2)}{3(n+q)\log_2 (n+q)} - \frac{n}{3(n+q)} \;  .
    $$
\end{theorem}
\begin{proof}
    Follows from the Theorem~\ref{theorem_complexity_lower}
    and the equation~\eqref{formula_depth_and_gate_complexity_dependency}.
\end{proof}

\begin{theorem}\label{theorem_quantum_weight_lower}
    There is such $n_0 \in \NN$ that for $n > n_0$ the following equation holds:
    $$
        W(n,q) \geq \min(\WC, \WT) \cdot \left(\frac{2^n(n-2)}{3\log_2 (n+q)} - \frac{n}{3} \right) \; .
    $$
\end{theorem}
\begin{proof}
    Follows from the Theorem~\ref{theorem_complexity_lower}
    and the equation~\eqref{formula_quantum_weight_and_gate_complexity_dependency}.
\end{proof}

\begin{theorem}\label{theorem_complexity_upper_no_memory}
    $$
        L(n, 0) \leqslant \frac{3n2^{n+4}}{\log_2 n - \log_2 \log_2 n - \log_2 \phi(n)}
            \left( 1 + \epsilon(n) \right) \;  ,
    $$
    where $\phi(n) < n \mathop / \log_2 n$ is an arbitrarily slowly growing function and $\epsilon(n)$ equals to:
    $$
        \epsilon(n) = \frac{1}{6\phi(n)} +\left(\frac{8}{3} - o(1)\right)
            \frac{\log_2 n \cdot \log_2 \log_2 n}{n} \;  .
    $$
\end{theorem}

\begin{theorem}\label{theorem_depth_upper_no_memory}
    $$
        D(n, 0) \leqslant \frac{n2^{n+5}}{\log_2 n - \log_2 \log_2 n - \log_2 \phi(n)}
            \left( 1 + \epsilon(n) \right) \;  ,
    $$
    where $\phi(n) < n \mathop / \log_2 n$ is an arbitrarily slowly growing function and $\epsilon(n)$ equals to:
    $$
        \epsilon(n) = \frac{1}{4\phi(n)} +(4 - o(1))\frac{\log_2 n \cdot \log_2 \log_2 n}{n} \;  .
    $$
\end{theorem}

\begin{theorem}\label{theorem_quantum_weight_upper_no_memory}
    $$
        W(n, 0) \leqslant \frac{n2^{n+4} \left( \WC(1 + \epsilon_C(n)) + 2\WT(1 + \epsilon_T(n) \right)}
            {\log_2 n - \log_2 \log_2 n - \log_2 \phi(n)} \;  ,
    $$
    where $\phi(n) < n \mathop / \log_2 n$ is an arbitrarily slowly growing function and:
    \begin{align*}
        \epsilon_C(n) &= \frac{1}{2\phi(n)} - \left( \frac{1}{2} - o(1) \right) \cdot \frac{\log_2 \log_2 n }{n} \;  , \\
        \epsilon_T(n) &= (4 - o(1))\frac{\log_2 n \cdot \log_2 \log_2 n}{n} \;  .
    \end{align*}
\end{theorem}

Proofs of the Theorems~\ref{theorem_complexity_upper_no_memory}--\ref{theorem_quantum_weight_upper_no_memory}
will be given in Section~\ref{section_upper_bound_no_memory}.

\begin{theorem}\label{theorem_complexity_main}
    $$
        L(n,0) \asymp n2^n \mathop / \log_2 n \;  .
    $$
\end{theorem}
\begin{proof}
    Follows from the Theorems~\ref{theorem_complexity_lower} and~\ref{theorem_complexity_upper_no_memory}.
\end{proof}

\section{General lower bounds}\label{section_lower_bounds}

As we said earlier, we can implement any permutation $h \in A(\ZZ_2^n)$ with a reversible circuit without additional inputs,
consisting of gates from $\Omega_n^2$. In paper~\cite{gluhov} it was proved that the length $L(G,M)$ of a permutation group $G$
with respect to a generating set $M$ has the following lower bound:
$$
    L(G,M) \geqslant \left \lceil \log_{|M|} |G| \right \rceil \;  .
$$
In our case we have $G = A(\ZZ_2^n)$, $|G| = (2^n)! \mathop / 2$, $|M| = |\Omega_n^2|$.
Since the cardinality of the set $\Omega_n^2$ equals to
\begin{equation}\label{formula_size_of_set_omega_n_2}
    |\Omega_n^2| = \sum_{k = 0}^2 {(n-k)  {n \choose k}} = \frac{n^3}{2} \left( 1 + o(1) \right)  \; ,
\end{equation}
we can derive a simple lower asymptotic bound for the function $L(n,0)$:
$$
    L(n,0) \gtrsim \frac{\log_2 ((2^n)! \mathop / 2)}{\log_2 (n^3 \mathop / 2)} \gtrsim \frac{n2^n}{3 \log_2 n}  \;  .  
$$
This bound is asymptotically equal to the bound $\Omega(n2^n \mathop / \log n)$ from the paper~\cite{shende}.

To derive a general lower bound for the function $L(n, q)$, we should take into account
all transformations $\ZZ_2^n \to \ZZ_2^n$ that can be implemented by a reversible circuit with $(n+q)$ inputs.
There are no more than $P(n+q,n)$ (an $n$-permutation of $(n+q)$) of such transformations.

Let's now proceed to the proof of the Theorem~\ref{theorem_complexity_lower}.
\begin{proof}[Proof of the Theorem~\ref{theorem_complexity_lower}]
    We use counting argument to prove that there is such $n_0 \in \NN$ that for $n > n_0$ the following equation holds:
    $$
        L(n,q) \geq \frac{2^n(n-2)}{3\log_2(n+q)} - \frac{n}{3} \;  .
    $$

    Let $r = |\Omega_n^2|$. From the equation~\eqref{formula_size_of_set_omega_n_2} it follows that
    \begin{gather*}
        r = \sum_{k=0}^2{(n-k)\binom{n}{k}} = \frac{n^3 - n^2 + 2n}{2} \;  , \notag \\
        \frac{n^2(n-1)}{2} + 1 < r \leq \frac{n^3}{2} \text{ \,, if\, } n \geq 2 \;  .
    \end{gather*}

    Let $\EuScript C^*(n,s)$ and $\EuScript C(n,s)$ be the number of all reversible circuits of gates
    from $\Omega_n^2$ with the gate complexity $s$ and no more than $s$ respectively.
    Then the following equations hold:
    \begin{align*}
        \EuScript C(n,s) &= \sum_{i=0}^s {\EuScript C^*(n,i)} = \frac{r^{s+1} - 1}{r-1}
            \leqslant \left( \frac{n^3}{2} \right)^{s+1} \cdot \frac{2}{n^2(n-1)}  \; , \\
        \EuScript C(n,s) &\leqslant \left( \frac{n^3}{2} \right)^s \cdot \left(1 + \frac{1}{n-1}\right)
            \text{ \,, if\, } n \geqslant 2 \;  .
    \end{align*}

    As we said earlier, there are no more than $P(n+q,n)$ of different transformations $\ZZ_2^n \to \ZZ_2^n$,
    that can be implemented by a reversible circuit with $(n+q)$ inputs.
    Hence, we can state that
    $$
       \EuScript C(n+q,L(n,q)) \cdot P(n+q,n) \geq |F(n,q)| \;  .
    $$
    Since $|F(n,q)| \geqslant |A(\ZZ_2^n)| = (2^n)! \mathop / 2$ and $P(n+q,n) \leq (n+q)^n$, it follows that
    $$
        \left( \frac{(n+q)^3}{2} \right)^{L(n,q)} \cdot \left(1 + \frac{1}{n+q-1}\right)
            \cdot (n+q)^n \geq (2^n)! \mathop / 2 \;  .
    $$

    There is such $n_0 \in \NN$ that for $n > n_0$ an equation $(2^n)! \geq (2^n \mathop / e)^{2^n}$ holds.
    For such values of $n$ we can state that
    \begin{multline*}
        L(n,q) \cdot (3\log_2(n+q) - 1) + \log_2 \left(1 + \frac{1}{n+q-1}\right) + \\
            + n \log_2(n+q) \geq 2^n(n - \log_2 e) \;  .
    \end{multline*}
    
    From this we obtain a general lower bound for the function $L(n,q)$:
    $$
        L(n,q) \geq \frac{2^n(n-2)}{3\log_2(n+q)} - \frac{n}{3} \;  .
    $$
\end{proof}

In the following section we will give a description of a new group theory based synthesis algorithm,
which can produce a reversible circuit with asymptotically the best gate complexity
and without additional inputs.

\section{Synthesis of circuits without additional inputs}\label{section_upper_bound_no_memory}

A reversible circuit without additional inputs, consisting of gates from $\Omega_n^2$, can implement only an even permutation.
In~\cite{my_fast_synthesis_algorithm} a group theory based synthesis algorithm was described.
This algorithm for any permutation $h \in A(\ZZ_2^n)$ can
produce a circuit $\frS$ implementing $h$ with the gate complexity $L(\frS) \lesssim 7n 2^n$.

Let us now describe a new synthesis algorithm which use a similar technique as the algorithm from%
~\cite{my_fast_synthesis_algorithm}, but has a better upper bound for the gate complexity of a synthesized circuit.
This algorithm's description will be given in a form of the theorem proof.

\begin{proof}[Proof of the Theorem~\ref{theorem_complexity_upper_no_memory}]
    We will describe a new group theory based synthesis algorithm,
    which for any permutation $h \in A(\ZZ_2^n)$ can produce a circuit $\frS$ implementing $h$ with the gate complexity
    $$
        L(\frS) \leq \frac{3n2^{n+4}}{\log_2 n - \log_2 \log_2 n - \log_2 \phi(n)}
            \left( 1 + \epsilon(n) \right)  \; ,
    $$
    where $\phi(n) < n \mathop / \log_2 n$ is an arbitrarily slowly growing function and the function $\epsilon(n)$ equals to
    $$
        \epsilon(n) = \frac{1}{6\phi(n)} +\left(\frac{8}{3} - o(1)\right)
            \frac{\log_2 n \cdot \log_2 \log_2 n}{n}
    $$

    Let's consider a permutation $h \in A(\ZZ_2^n)$ and the transformation $f_h\colon \ZZ_2^n \to \ZZ_2^n$, defined by it.
    The main idea is in a decomposition of $h$ into the product
    of transpositions in such a way that all of them can be grouped by $K$ independent transpositions:
    \begin{equation}\label{formula_permutation_decomposition_main}
        h = G_1 \circ G_2 \circ \cdots \circ G_t \circ h' \;  ,    
    \end{equation}
    where $G_i = (\vv x_{i,1}, \vv y_{i,1}) \circ \cdots \circ (\vv x_{i,K}, \vv y_{i,K})$ is an $i$-th group
    of $K$ independent transpositions, $\vv x_{i,j}, \vv y_{i,j} \in \ZZ_2^n$ and $h'$ is a residual permutation.
    Our goal is to show that a group $G_i$ can be implemented by a composition of a single multiple controls Toffoli gate
    and many CNOT and 2-CNOT gates.

    The permutation $h$ can also be represented as the product of independent cycles
    with the sum of cycles lengths no more than $2^n$. Having this permutation representation,
    we can obtain independent transpositions from the cycles as follows:
    \begin{multline}\label{formula_decompostion_of_two_cycles}
        (i_1, \cdots, i_r) \circ (j_1, \cdots, j_s) = (i_1, i_2) \circ (j_1, j_2) \circ \\
            \circ (i_1, i_3, \cdots, i_r) \circ (j_1, j_3, \cdots, j_s) \;  .
    \end{multline}
    \begin{equation}\label{formula_decompostion_of_k_cycle}
        (i_1, \cdots, i_r) = (i_1, i_2) \circ (i_3, i_4) \circ (i_1, i_3, i_5, i_6, \cdots, i_r)
            \text{\; for \;} r \geq 5 \;  .
    \end{equation}

    If we look at the permutation representation~\eqref{formula_permutation_decomposition_main}
    and the equations \eqref{formula_decompostion_of_two_cycles}--\eqref{formula_decompostion_of_k_cycle},
    we will see that $K$ independent transpositions can't be obtained from a residual permutation $h'$
    only if it has less than $K$ independent cycles and every of these cycles has the length less than 5.
    Hence, a sum of the cycles lengths of the permutation $h'$ is no more than $4(K-1)$.

    Let $M_g$ be the set of non-fixed points of a permutation $g \in S(\ZZ_2^n)$:
    $$
        M_g = \{\,\vv x \in \ZZ_2^n\mid g(\vv x) \ne \vv x\,\} \;  .
    $$
    Then we can state that $|M_h| \leq 2^n$, $|M_{h'}| \leq 4(K-1)$.

    If we apply the equation~\eqref{formula_permutation_decomposition_main} to the permutation $h'$ providing that $K = 2$,
    we will see that this permutation can be represented as a product of no more than $|M_{h'}| \mathop / 2$
    independent transpositions pairs and one cycle with the length 3 at most.
    Every cycle with the length 3 also can be represented as a product of two independent transpositions pairs:
    $$
        (i, j, k) = ((i,j) \circ (r, s)) \circ ((r,s) \circ (i,k)) \;  .
    $$

    Let $g^{(i)}$ be a permutation which is represented as a product of $i$ independent transpositions and
    $f_{g^{(i)}}$ be the transformation defined by this permutation.
    Now we can derive an upper bound for the function $L(f_h,0)$, where $f_h$ is the transformation,
    defined by a permutation $h \in A(\ZZ_2^n)$:
    \begin{gather}
        L(f_h,0) \leq \frac{|M_h|}{K} \cdot L(f_{g^{(K)}},0)
            + \left(\frac{|M_{h'}|}{2} + 2 \right ) \cdot L(f_{g^{(2)}},0) \;  , \notag \\
        L(f_h,0) \leq \frac{2^n}{K} \cdot L(f_{g^{(K)}},0) + 2K \cdot L(f_{g^{(2)}},0) \;  .
        \label{formula_upper_bound_of_L_h_common}
    \end{gather}
    All we should do now is to find out an upper bound for the function $L(f_{g^{(K)}},0)$.

    Let's consider an arbitrary permutation $g^{(K)} \in A(\ZZ_2^n)$.
    Let $k$ be the cardinality of the set $M_{g^{(K)}}$: $k = |M_{g^{(K)}}|$, then $k = 2K$.
    The essence of the proposed synthesis algorithm is in a conjugation of the permutation $g^{(K)}$
    in order to get a permutation corresponding to a single $k$-CNOT gate.
    Every gate $e$ from $\Omega_n^2$ defines a permutation $h_e$, for which $h^{-1}_e = h_e$.
    This means that conjugating $g^{(K)}$ by $h_e$ corresponds to the attaching the gate $e$
    to the front and back of a current sub-circuit.

    Let $g^{(K)} = (\vv x_1, \vv y_1) \circ \cdots \circ (\vv x_K, \vv y_K)$.
    We define a matrix $A$ as follows:
    \begin{equation}
        A =
            \left(
                \begin{matrix}
                    \vv x_1 \\
                    \vv y_1 \\
                    \cdots \\
                    \vv x_K \\
                    \vv y_K
                \end{matrix}
            \right )
          =
            \left(
                \begin{matrix}
                    a_{1,1}   & \cdots & a_{1,n}   \\
                    a_{2,1}   & \cdots & a_{2,n}   \\
                    \hdotsfor{3}                   \\
                    a_{k-1,1} & \cdots & a_{k-1,n} \\
                    a_{k,1}   & \cdots & a_{k,n}   \\
                \end{matrix}
            \right ) \;  .
        \label{formula_matrix_for_permutation}
    \end{equation}

    Let $k$ be the power of two: $2^{\lfloor \log_2 k \rfloor} = k$.
    If $k \leq \log_2 n$, then we can state that no more than $2^k$ and no less than $\log_2 k$
    pairwise distinct columns exist in the matrix $A$.
    Without the loss of generality we can assume that all $d \leq 2^k$ pairwise distinct columns are the first ones.
    Then for every $j$-th column, $j > d$, there is equal to it an $i$-th column, $i \leq d$.
    If we conjugate $g^{(K)}$ by the permutation, corresponding to the gate $C_{i;j}$,
    we will zero out a $j$-th column in the matrix $A$.
    We do this for all the columns whose index is greater than $d$ using $L_1 = 2(n-d)$ CNOT gates.%
    \phantomsection\label{page_l1}
    In result we obtain a new permutation $g_1^{(K)}$ and a new matrix $A_1$ for it as follows:
    $$
        A_1 =
            \left(\phantom{\begin{matrix} 0\\0\\0\\0\\0\\ \end{matrix}} \right.
            \begin{matrix}
                a_{1,1}   & \cdots & a_{1,d}   \\
                a_{2,1}   & \cdots & a_{2,d}   \\
                \hdotsfor{3}                     \\
                a_{k-1,1} & \cdots & a_{k-1,d} \\
                a_{k,1}   & \cdots & a_{k,d}   \\
            \end{matrix}
            \phantom{\begin{matrix} 0\\0\\0\\0\\0\\ \end{matrix}}
            \overbrace{
                \begin{matrix}
                    0 &\cdots & 0   \\
                    0 &\cdots & 0   \\
                    \hdotsfor{3}    \\
                    0 &\cdots & 0   \\
                    0 &\cdots & 0   \\
                \end{matrix}
            }^{n - d}        
            \left.\phantom{\begin{matrix} 0\\0\\0\\0\\0\\ \end{matrix}} \right) \;  .
    $$

    Now for every $a_{1,i} = 1$ we conjugate $g_1^{(K)}$ by the permutation corresponding to the gate $N_i$
    in order to zero out the first row of the matrix $A_1$. We need $L_2 = 2d$ NOT gates to do this.%
    \phantomsection\label{page_l2}
    In result we obtain a new permutation $g_2^{(K)}$ and a new matrix $A_2$ for it as follows:
    $$
        A_2 =
            \left(\phantom{\begin{matrix} 0\\0\\0\\0\\0\\ \end{matrix}} \right.
            \begin{matrix}
                0         & \cdots & 0           \\
                b_{2,1}   & \cdots & b_{2,d}   \\
                \hdotsfor{3}                     \\
                b_{k-1,1} & \cdots & b_{k-1,d} \\
                b_{k,1}   & \cdots & b_{k,d}   \\
            \end{matrix}
            \phantom{\begin{matrix} 0\\0\\0\\0\\0\\ \end{matrix}}
            \overbrace{
                \begin{matrix}
                    0 &\cdots & 0   \\
                    0 &\cdots & 0   \\
                    \hdotsfor{3}    \\
                    0 &\cdots & 0   \\
                    0 &\cdots & 0   \\
                \end{matrix}
            }^{n - d}        
            \left.\phantom{\begin{matrix} 0\\0\\0\\0\\0\\ \end{matrix}} \right) \;  .
    $$

    Next step is reducing the matrix $A_2$ to a \textit{canonical form},
    where every row, after reversing the order of its elements, will represent itself the binary expansion
    of row index minus 1.
    
    All the rows in the matrix $A_2$ are distinct. The first row is already in the canonical form,
    so we will successively transform the rest of the rows beginning from the second one.
    Let's assume that the current row has an index $i$ and all the rows with indices
    from $1$ to $(i-1)$ are in the canonical form. There are two cases:
    \begin{enumerate}
        \item
            There is nonzero element $b_{i,j}$ in the $i$-th row with an index $j > \log_2 k$.
            In this case for every element $b_{i, j'}$, $j' \ne j$, $j' < d$,
            which is not equal to the $j'$-th bit in the binary expansion of the number $(i-1)$,
            we conjugate $g_2^{(K)}$ by the permutation corresponding to the gate $C_{j;j'}$.
            This will require no more than $2d$ CNOT gates.
            To make the current row canonical, we should now zero out only the $j$-th element of it.
            It can be done with conjugating $g_2^{(K)}$ by the permutation corresponding to the gate $C_{I;j}$,
            where $I$ is the set of nonzero bits indices in the binary expansion of the number $(i-1)$.
            For example, if $i = 6$, then $I = \{\,1, 3\,\}$.
            Since $|I| \leq \log_2 k$, we can replace this multiple controls Toffoli gate by a composition of no more than
            $8 \log_2 k$ Toffoli gates~\cite{barenco}, thus we need no more than $16 \log_2 k$ Toffoli gates for this part.
            
            So, summing up, in this case we need $L_3^{(i)} \leq 2d + 16 \log_2 k$ gates from $\Omega_n^2$
            to transform the $i$-th row to the canonical form.

        \item
            There is no nonzero element in the $i$-th row with an index $j > \log_2 k$: $b_{i,j} = 0$
            for all $j > \log_2 k$.
            In this case, we conjugate $g_2^{(K)}$ by the permutation corresponding to the gate $C_{I;\log_2 k + 1}$,
            where $I$ is the set of current row nonzero elements' indices.
            Because of inequality of matrix rows and because all the previous rows are in the canonical form,
            we can state that the value of $b_{j,\log_2 k + 1}$ will be inverted only if $j \geq i$.
            Since $|I| \leq \log_2 k$, we can replace this multiple controls Toffoli gate by a composition of no more than
            $8 \log_2 k$ Toffoli gates~\cite{barenco}, thus we need no more than $16 \log_2 k$ Toffoli gates for this part.
            After that we can go to the previous case.
            
            So, summing up, in this case we need $L_3^{(i)} \leq 2d + 32\log_2 k$
            gates from $\Omega_n^2$ to transform the $i$-th row to the canonical form.
    \end{enumerate}
    
    As we can see, we obtained a new restriction to the value of $k$: $\log_2 k$ should be strictly less than $n$,
    otherwise we will not be able to transform the matrix $A_2$ to the canonical form.
    After this transforming, we obtain a new permutation $g_3^{(K)}$
    and a new matrix $A_3$ for it as follows:
    $$
        A_3 =
            \left(\phantom{\begin{matrix} 0\\0\\0\\0\\0\\ \end{matrix}} \right.
            \overbrace{
                \begin{matrix}
                    0 & 0 & 0 & \cdots & 0 \\
                    1 & 0 & 0 & \cdots & 0 \\
                    \hdotsfor{5}           \\
                    0 & 1 & 1 & \cdots & 1 \\
                    1 & 1 & 1 & \cdots & 1 \\
                \end{matrix}
            }^{\log_2 k}        
            \phantom{\begin{matrix} 0\\0\\0\\0\\0\\ \end{matrix}}
            \overbrace{
                \begin{matrix}
                    0 &\cdots & 0   \\
                    0 &\cdots & 0   \\
                    \hdotsfor{3}    \\
                    0 &\cdots & 0   \\
                    0 &\cdots & 0   \\
                \end{matrix}
            }^{n - \log_2 k}        
            \left.\phantom{\begin{matrix} 0\\0\\0\\0\\0\\ \end{matrix}} \right) \;  .
    $$
    For this transformation we need $L_3$ gates from $\Omega_n^2$\phantomsection\label{page_l3}:
    $$
        L_3 = \sum_{i=2}^k {L_3^{(i)}} \leq k(2d + 32 \log_2 k) \;  .
    $$

    Finally, for every $i > \log_2 k$ we conjugate $g_3^{(K)}$ by the permutation corresponding to the gate $N_i$.
    We need $L_4 = 2(n - \log_2 k)$ NOT gates to do this.\phantomsection\label{page_l4}
    In result we obtain a new permutation $g_4^{(K)}$ and a new matrix $A_4$ for it as follows:
    $$
        A_4 =
            \left(\phantom{\begin{matrix} 0\\0\\0\\0\\0\\ \end{matrix}} \right.
            \overbrace{
                \begin{matrix}
                    0 & 0 & 0 & \cdots & 0 \\
                    1 & 0 & 0 & \cdots & 0 \\
                    \hdotsfor{5}           \\
                    0 & 1 & 1 & \cdots & 1 \\
                    1 & 1 & 1 & \cdots & 1 \\
                \end{matrix}
            }^{\log_2 k}        
            \phantom{\begin{matrix} 0\\0\\0\\0\\0\\ \end{matrix}}
            \overbrace{
                \begin{matrix}
                    1 &\cdots & 1   \\
                    1 &\cdots & 1   \\
                    \hdotsfor{3}    \\
                    1 &\cdots & 1   \\
                    1 &\cdots & 1   \\
                \end{matrix}
            }^{n - \log_2 k}        
            \left.\phantom{\begin{matrix} 0\\0\\0\\0\\0\\ \end{matrix}} \right) \;  .
    $$

    The permutation $g_4^{(K)}$ corresponds to the single gate $C_{n,n-1, \cdots, \log_2 k + 1; 1}$. This gate has
    $(n - \log_2 k)$ control inputs, thus it can be replaced by no more than $L_5 \leq 8(n - \log_2 k)$%
    \phantomsection\label{page_l5} Toffoli gates~\cite{barenco}.

    We obtained the permutation $g_4^{(K)}$ with the help of conjugation the permutation $g^{(K)}$ by specific permutations.
    If we conjugate $g_4^{(K)}$ by exactly the same permutations, but in a reverse order,
    we will obtain $g^{(K)}$. In terms of a circuit synthesis this means that we
    should attach all the gates we used in our matrix transformations to the gate $C_{n,n-1, \cdots, \log_2 k + 1; 1}$
    from left and right, but in a reverse order. As a result we will obtain a reversible circuit $\frS_K$,
    which defines the permutation $g^{(K)}$.
    From this it follows that $L(g^{(K)},0) \leq L(\frS_K)$ and
    \begin{multline*}
        L(g^{(K)},0) \leq \sum_{i=1}^5 {L_i} \leq 2(n-d) + 2d + \\
            + k(2d + 32 \log_2 k) + 2(n-\log_2 k) + 8(n - \log_2 k) \;  ,
    \end{multline*}
    $$
        L(g^{(K)},0) \leq 12n + k2^{k+1} + 32k\log_2 k - 10 \log_2 k \; .
    $$
    Also, $L(g^{(2)},0) \leq 12n + 364$.
    
    Using these upper bounds in the equation~\eqref{formula_upper_bound_of_L_h_common},
    we obtain the following upper bound for the function $L(f_h,0)$:
    $$
        L(f_h,0) \leq \frac{2^{n+1}}{k}(12n + k2^{k+1} + 32k\log_2 k - 10 \log_2 k) + k(12n + 364) \;  .
    $$

    Our synthesis algorithm requires $k$ to be the power of two and $\log_2 k$ to be strictly less than $n$.
    Let $m = \log_2 n - \log_2 \log_2 n - \log_2 \phi(n)$ and $k = 2^{\lfloor \log_2 m \rfloor}$,
    where $\phi(n) < n \mathop / \log_2 n$ is an arbitrarily slowly growing function.
    Then $m / 2 \leq k \leq m$ and
    \begin{gather*}
        L(f_h,0) \leq \frac{2^{n+2}}{m}(12n + 2m2^m + (32 - o(1))m\log_2 m) \;  , \\
        L(f_h,0) \leq \frac{3n2^{n+4}}{m} 
            \left( 1 + \frac{2^m\log_2 n}{6n} + \left(\frac{8}{3} - o(1)\right)
            \frac{\log_2 n \cdot \log_2 \log_2 n}{n} \right) \;  .
    \end{gather*}

    From this we obtain the final upper bound for the function $L(f_h,0)$:
    \begin{equation}\label{formula_gate_complexity_of_transformation_no_memory}
        L(f_h,0) \leq \frac{3n2^{n+4}}{\log_2 n - \log_2 \log_2 n - \log_2 \phi(n)}
            \left( 1 + \epsilon(n) \right) \;  ,
    \end{equation}
    where the function $\epsilon(n)$ equals to
    $$
        \epsilon(n) = \frac{1}{6\phi(n)} +\left(\frac{8}{3} - o(1)\right)
            \frac{\log_2 n \cdot \log_2 \log_2 n}{n} \;  .
    $$

    Since our synthesis algorithm can produce a reversible circuit $\frS$ for an arbitrary permutation $h \in A(\ZZ_2^n)$,
    it follows that the function $L(n,0)$ is upper bounded by the same value as $L(f_h,0)$.
\end{proof}

To explain the main part of our synthesis algorithm, let's consider a permutation
$g^{(2)} = (\langle 1,0,0,1 \rangle, \langle 0, 0, 0, 0 \rangle) \circ (\langle 1,1,1,1 \rangle, \langle 0, 1, 1, 0 \rangle)$.
This permutation can be implemented by a reversible circuit
$\frS = C_{1;4} * C_{2;3} * N_1 * N_3 * N_4 * C_{3,4;1} * N_4 * N_3 * N_1 * C_{2;3} * C_{1;4}$.
The process of obtaining the circuit $\frS$ is showed in Fig.~\ref{pic_conjugation_process}.

\begin{figure}[ht]
    \centering
    \begin{tabular}{ccc}
        $\frS$ &
        $C_{1;4} * \frS * C_{1;4}$ &
        $C_{2;3} * \frS_1 * C_{2;3}$
        \\
        
        \\
        
        $A =
        \left(
            \begin{matrix}
                1 & 0 & 0 & 1 \\
                0 & 0 & 0 & 0 \\
                1 & 1 & 1 & 1 \\
                0 & 1 & 1 & 0 \\
            \end{matrix}
        \right ) \Rightarrow$ &
        $\left(
            \begin{matrix}
                1 & 0 & 0 & 0 \\
                0 & 0 & 0 & 0 \\
                1 & 1 & 1 & 0 \\
                0 & 1 & 1 & 0 \\
            \end{matrix}
        \right ) \Rightarrow$ &
        $\left(
            \begin{matrix}
                1 & 0 & 0 & 0 \\
                0 & 0 & 0 & 0 \\
                1 & 1 & 0 & 0 \\
                0 & 1 & 0 & 0 \\
            \end{matrix}
        \right ) \Rightarrow$
    \end{tabular}
    
    \bigskip
    \bigskip

    \begin{tabular}{ccc}
        $N_1 * \frS_2 * N_1$ &
        $N_3 * \frS_3 * N_3$ &
        $N_4 * \frS_4 * N_4$
        \\
        
        \\
        
        $\Rightarrow
        \left(
            \begin{matrix}
                0 & 0 & 0 & 0 \\
                1 & 0 & 0 & 0 \\
                0 & 1 & 0 & 0 \\
                1 & 1 & 0 & 0 \\
            \end{matrix}
        \right )$ &
        $\Rightarrow
        \left(
            \begin{matrix}
                0 & 0 & 1 & 0 \\
                1 & 0 & 1 & 0 \\
                0 & 1 & 1 & 0 \\
                1 & 1 & 1 & 0 \\
            \end{matrix}
        \right )$ &
        $\Rightarrow
        \left(
            \begin{matrix}
                0 & 0 & 1 & 1 \\
                1 & 0 & 1 & 1 \\
                0 & 1 & 1 & 1 \\
                1 & 1 & 1 & 1 \\
            \end{matrix}
        \right )$
    \end{tabular}

    \smallskip
    \bigskip

    $
        N_4 * N_3 * N_1 * C_{2;3} * C_{1;4} * \frS * C_{1;4} * C_{2;3} * N_1 * N_3 * N_4 = C_{3,4;1}
    $
    \caption
    {
        The process of obtaining a reversible circuit $\frS$, implementing a permutation
        $g^{(2)} = (\langle 1,0,0,1 \rangle, \langle 0, 0, 0, 0 \rangle)
        \circ (\langle 1,1,1,1 \rangle, \langle 0, 1, 1, 0 \rangle)$.
    }\label{pic_conjugation_process}
\end{figure}

The proposed synthesis algorithm allows us to prove the Theorem~\ref{theorem_depth_upper_no_memory}.
\begin{proof}[Proof of the Theorem~\ref{theorem_depth_upper_no_memory}]
    We should prove the following equation:
    $$
        D(n, 0) \leq \frac{n2^{n+5}}{\log_2 n - \log_2 \log_2 n - \log_2 \phi(n)}
            \left( 1 + \epsilon(n) \right) \;  ,
    $$
    where $\phi(n) < n \mathop / \log_2 n$ is an arbitrarily slowly growing function and $\epsilon(n)$ equals to
    $$
        \epsilon(n) = \frac{1}{4\phi(n)} + (4 - o(1))\frac{\log_2 n \cdot \log_2 \log_2 n}{n} \;  .
    $$
    This can be easily done, if we take into account that some of operations in the proposed synthesis algorithm can be
    done with the logarithmic depth. For example, we can zero out duplicating columns in the matrix with the logarithmic depth
    (see Fig.~\ref{pic_zeroing_out}).
    Also, a conjugation by permutations, corresponding to NOT gates, can be done with the constant depth.
    
    \begin{figure}[ht]
        \centering
        \includegraphics[scale=1.2]{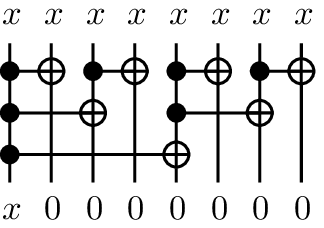}
        \caption
        {
            Clearing duplicating inputs with the logarithmic depth.
        }\label{pic_zeroing_out}
    \end{figure}

    This implies that $D_1 = 2 \lceil \log_2(n-d) \rceil$ (against $L_1 = 2(n-d)$, see page~\pageref{page_l1}),
    $D_2 = 2$ (against $L_2 = 2d$, see page~\pageref{page_l2}) and $D_4 = 2$
    (against $L_4 = 2(n-\log_2 k)$, see page~\pageref{page_l4}).
    All other parts of our synthesis algorithm produce sub-circuits with the depth equal to the gate complexity:
    $D_3 = L_3 \leq k(2d + 32 \log_2 k)$, $D_5 = L_5 \leq 8(n - \log_2 k)$
    (see page~\pageref{page_l3}).

    Using these depth values, we can derive the following upper bound:
    $$
        D(g^{(K)},0) \leq \sum_{i=1}^5 {D_i} \leq 2 \log_2 n + k(2^{k+1} + 32 \log_2 k)
            + 8(n - \log_2 k) + 6 \;  .
    $$
    Also, $D(g^{(2)},0) \leq 8n + 2\log_2 n + 374$.

    Using these upper bounds in the equation~\eqref{formula_upper_bound_of_L_h_common} for the circuit depth,
    we obtain the following upper bound for the function $D(f_h,0)$:
    $$
        D(f_h,0) \leq \frac{2^{n+1}}{k}(8n + 2\log_2 n + k2^{k+1} + (32 - o(1))k\log_2 k) \;  .
    $$

    Providing $m = \log_2 n - \log_2 \log_2 n - \log_2 \phi(n)$ and $k = 2^{\lfloor \log_2 m \rfloor}$,
    where $\phi(n) < n \mathop / \log_2 n$ is an arbitrarily slowly growing function,
    we obtain an upper bound for the function $D(f_h,0)$:
    \begin{gather*}
        D(f_h,0) \leq \frac{n2^{n+5}}{m} 
            \left( 1 + \frac{\log_2 n}{4n} + \frac{m 2^m}{4n} + (4 - o(1))
            \frac{m \log_2 m}{n} \right) \;  , \\
        D(f_h,0) \leq \frac{n2^{n+5}}{\log_2 n - \log_2 \log_2 n - \log_2 \phi(n)}
            \left( 1 + \epsilon(n) \right) \;  ,
    \end{gather*}
    where the function $\epsilon(n)$ equals to
    $$
        \epsilon(n) = \frac{1}{4\phi(n)} + (4 - o(1)) \frac{\log_2 n \cdot \log_2 \log_2 n}{n} \;  .
    $$

    Since our synthesis algorithm can produce a reversible circuit $\frS$ for an arbitrary permutation $h \in A(\ZZ_2^n)$,
    it follows that the function $D(n,0)$ is upper bounded by the same value as $D(f_h,0)$.
\end{proof}

Now we should prove the last theorem of this paper.
\begin{proof}[Proof of the Theorem~\ref{theorem_quantum_weight_upper_no_memory}]
    We should count the number of NOT, CNOT and 2-CNOT gates in a reversible circuit,
    synthesized by the proposed synthesis algorithm, to prove an upper bound of the theorem
    $$
        W(n, 0) \leqslant \frac{n2^{n+4} \left( \WC(1 + \epsilon_C(n)) + 2\WT(1 + \epsilon_T(n) \right)}
            {\log_2 n - \log_2 \log_2 n - \log_2 \phi(n)} \;  ,
    $$
    where $\phi(n) < n \mathop / \log_2 n$ is an arbitrarily slowly growing function and
    \begin{align*}
        \epsilon_C(n) &= \frac{1}{2\phi(n)} - \left( \frac{1}{2} - o(1) \right) \cdot \frac{\log_2 \log_2 n }{n} \;  , \\
        \epsilon_T(n) &= (4 - o(1))\frac{\log_2 n \cdot \log_2 \log_2 n}{n} \;  .
    \end{align*}
    
    We can see that
    \begin{align*}
        \LC_1 & = 2(n-d)  \; \text{(CNOT gates only)} \; ,        & \LT_1 & = 0 \; ; \\
        \LC_2 & = 2d  \; \text{(NOT gates only)} \; ,             & \LT_2 & = 0 \; ; \\
        \LC_3 & \leq 2kd  \; \text{(CNOT gates only)} \; ,        & \LT_3 & \leq 32k \log_2 k \; ; \\
        \LC_4 & = 2(n - \log_2 k)  \; \text{(NOT gates only)} \; ,& \LT_4 & = 0 \; ; \\
        \LC_5 & = 0 \; ,                                          & \LT_5 & \leq 8(n - \log_2 k) \; .
    \end{align*}
    
    Summing up, we obtain the following upper bounds:
    \begin{gather*}
        \LC(g^{(K)},0) \leq \sum_{i=1}^5 {\LC_i} \leq 2(n-d) + 2d + 2kd + 2(n - \log_2 k) \; ,\\
        \LT(g^{(K)},0) \leq \sum_{i=1}^5 {\LT_i} \leq 32k \log_2 k + 8(n - \log_2 k) \;  , \\
        \LC(g^{(K)},0) \leq 4n + k 2^{k+1} - 2 \log_2 k \; , \\
            \LT(g^{(K)},0) \leq 8n + 32k \log_2 k - 8 \log_2 k \;  .
    \end{gather*}
    Also, $\LC(g^{(2)},0) \leq 4n + 124$ and $\LT(g^{(2)},0) \leq 8n + 240$.

    Using the equation~\eqref{formula_upper_bound_of_L_h_common},
    we obtain the following upper bounds for the functions $\LC(f_h,0)$ and $\LT(f_h,0)$:
    \begin{gather*}
        \LC(f_h,0) \leq \frac{2^{n+1}}{k} (4n + k 2^{k+1} - 2 \log_2 k) + k(4n + 124) \;  ,\\
        \LT(f_h,0) \leq \frac{2^{n+1}}{k} (8n + 32k \log_2 k - 8 \log_2 k) + k(8n + 240) \;  .
    \end{gather*}

    Providing $m = \log_2 n - \log_2 \log_2 n - \log_2 \phi(n)$ and $k = 2^{\lfloor \log_2 m \rfloor}$,
    where $\phi(n) < n \mathop / \log_2 n$ is an arbitrarily slowly growing function,
    we obtain the following upper bounds for the functions $\LC(f_h,0)$ and $\LT(f_h,0)$:
    \begin{gather*}
        \LC(f_h,0) \leq \frac{n2^{n+4}}{\log_2 n - \log_2 \log_2 n - \log_2 \phi(n)}
            \left( 1 + \epsilon_C(n) \right) \;  ,\\
        \LT(f_h,0) \leq \frac{n2^{n+5}}{\log_2 n - \log_2 \log_2 n - \log_2 \phi(n)}
            \left( 1 + \epsilon_T(n) \right) \;  ,
    \end{gather*}
    where the functions $\epsilon_C(n)$ and $\epsilon_T(n)$ equal to
    \begin{align*}
        \epsilon_C(n) &= \frac{1}{2\phi(n)} - \left( \frac{1}{2} - o(1) \right) \cdot \frac{\log_2 \log_2 n }{n} \;  , \\
        \epsilon_T(n) &= (4 - o(1))\frac{\log_2 n \cdot \log_2 \log_2 n}{n} \;  .
    \end{align*}
        
    Since our synthesis algorithm can produce a reversible circuit $\frS$ for an arbitrary permutation $h \in A(\ZZ_2^n)$,
    it follows that the function $W(n,0)$ is upper bounded by the same value as $W(f_h,0)$.
    From upper bounds for the functions $\LC(f_h,0)$ and $\LT(f_h,0)$ and
    from the equation~\eqref{formula_quantum_weight_and_gate_complexity_dependency},
    an upper bound for the function $W(n,0)$ from the Theorem~\ref{theorem_quantum_weight_upper_no_memory} follows.
\end{proof}

From the proof it follows that the ratio of the numbers of gates NOT, CNOT and Toffoli in a synthesized circuit
is approximately equal to 1:1:4.

\section{Conclusion}

We have discussed the problem of general synthesis of a reversible circuit without additional inputs,
consisting of NOT, CNOT and 2-CNOT gates, with the lowest possible gate complexity and depth.
We have studied the Shannon gate complexity function $L(n,q)$, the depth function $D(n,q)$ and
the quantum weight function $W(n,q)$ for a reversible circuit, implementing a transformation $f\colon \ZZ_2^n \to \ZZ_2^n$
from the set $F(n,q)$ without additional inputs.

From the lower bounds of these function we can see that using additional inputs should reduce 
the circuit's gate complexity and the depth. This is in line with respective practical evaluations as e.g. conducted in%
~\cite{miller_reducing_complexity, abdessaied_reducing_depth}.
Also, in paper~\cite{my_complexity_with_memory} an upper asymptotic bound $2^n$ for the function $L(n,q)$
in case of using additional inputs was established.
This bound is asymptotically lower than our bound for $L(n,0)$, but a significant number of additional inputs
in a reversible circuit is required to achieve it.

When solving the problem of reversible logic synthesis one should find a compromise between the gate complexity,
the depth (working time) and the amount of used memory (additional inputs) of a reversible circuit.
Further research should establish a more precise relationship of these parameters from each other.


\end{document}